\documentclass[reqno]{amsart}
\usepackage[utf8]{inputenc}
\usepackage{amsmath}
\usepackage{amssymb}
\usepackage{xcolor}
\usepackage{graphicx}
\usepackage{float}
\usepackage[hidelinks]{hyperref}
\usepackage{cite}
\newtheorem{theorem}{Theorem}

\newtheorem*{main-result}{Main result}

\newcommand{\C}{\mathbb{C}}
\newcommand{\R}{\mathbb{R}}

\newcommand{\om}{\omega}
\newcommand{\pur}{\mathcal{P}_1(\mathcal{H})}

\newcommand{\ler}[1]{\left( #1 \right)}
\newcommand{\lesq}[1]{\left[ #1 \right]}

\newcommand{\lers}[1]{\left\{ #1 \right\}}
\newcommand{\hohc}{\cH \otimes \cH^*}

\newcommand{\abs}[1]{\left| #1 \right|}

\newcommand{\inner}[2]{\left< #1,#2 \right>}

\newcommand{\be}{\begin{equation}}
\newcommand{\ee}{\end{equation}}
\newcommand{\ba}{\begin{array}}
\newcommand{\ea}{\end{array}}

\newcommand{\fel}{\frac{1}{2}}
\newcommand{\cA}{\mathcal{A}}

\newcommand{\cH}{\mathcal{H}}

\newcommand{\cP}{\mathcal{P}}

\newcommand{\cS}{\mathcal{S}}
\newcommand{\cC}{\mathcal{C}}
\newcommand{\cU}{\mathcal{U}}

\newcommand{\cL}{\mathcal{L}}
\newcommand{\bb}{\mathbf{b}}

\newcommand{\bO}{\mathbf{O}}

\newcommand{\tr}{\mathrm{tr}}

\begin{document}
\title[Qubit Wasserstein isometries]{Isometries of the qubit state space with respect to quantum Wasserstein distances}

\author[Rich\'ard Simon]{Rich\'ard Simon}
\address{Rich\'ard Simon, Department of Analysis, Institute of Mathematics, Budapest University of Technology and Economics, Műegyetem rkp. 3., H-1111 Budapest, Hungary.}
\email{sr570063@gmail.com; simonr@math.bme.hu}

\author[D\'aniel Virosztek]{D\'aniel Virosztek}
\address{D\'aniel Virosztek, Alfr\'ed R\'enyi Institute of Mathematics\\ Re\'altanoda u. 13-15.\\Budapest H-1053\\ Hungary}
\email{virosztek.daniel@renyi.hu}


\thanks{R. Simon is supported by the New National Excellence Program (grant no. \'UNKP-23-3). D. Virosztek is supported by the Momentum Program of the Hungarian Academy of Sciences under grant agreement no. LP2021-15/2021, and partially supported by the ERC Synergy Grant no. 810115.}


\maketitle
\begin{abstract}
In this paper we study isometries of quantum Wasserstein distances and divergences on the quantum bit state space. We describe isometries with respect to the symmetric quantum Wasserstein divergence $d_{\text{sym}}$, the divergence induced by all of the Pauli matrices. We also give a complete characterization of isometries with respect to $D_z$, the quantum Wasserstein distance corresponding to the single Pauli matrix $\sigma_z$.
\end{abstract}
\section{Introduction}
\subsection{Motivation and main results}

Even though the classical problem of transporting mass in an optimal way was already formulated in the 18th century by Monge, the theory of classical optimal transport (OT) became an intensively researched topic of analysis only in the last couple of decades with strong relations to mathematical physics \cite{DL-DP-M-T,JKO-98,Lott}, the theory of partial differential equations \cite{Ambrosio-lecturenotes,Figalli-note,OPS}, and probability \cite{BP,BiegelbockJuillet,HMS}. One could also mention the several applications in artificial intelligence, image processing, and various other fields of applied sciences. See e.g. \cite{Peyre1,Peyre2,Santambrogio} and the references therein.
\par
The quantum counterpart of the classical theory has just emerged in the recent years. As usual, the correspondence between the classical world and the quantum world is not one-to-one. The non-commutative counterpart of optimal transport theory is a flourishing research field these days with various different approaches such as that of Biane and Voiculescu \cite{BianeVoiculescu}, Carlen, Maas, Datta, and Rouzé  \cite{CarlenMaas-1,CarlenMaas-2,CarlenMaas-3,DattaRouze1,DattaRouze2}, Caglioti, Golse, Mouhot, and Paul \cite{CagliotiGolsePaul, CagliotiGolsePaul-towardsqot, GolseMouhotPaul, GolsePaul-wavepackets, GolsePaul-Schrodinger,GolsePaul-meanfieldlimit}, De Palma and Trevisan \cite{DePalmaMarvianTrevisanLloyd,DPT-AHP, DPT-lecture-notes}, \.Zyczkowski and his collaborators  \cite{FriedmanEcksteinColeZyczkowski-MK,ZyczkowskiSlomczynski1,ZyczkowskiSlomczynski2,BistronEcksteinZyczkowski}, and Duvenhage \cite{Duvenhage1, Duvenhage2}. Separable quantum Wasserstein distances have also been introduced recently \cite{TothPitrik}.
\par
In this paper, we follow the approach  of De Palma and Trevisan involving quantum channels, which is closely related to the quantum optimal transport concept of Caglioti, Golse, Mouhot, and Paul based on quantum couplings.
A common feature of both the channel-based De Palma-Trevisan approach and the coupling-based Caglioti-Golse-Mouhot-Paul approach is that the induced optimal transport distances are not genuine metrics. In particular, states can have strictly positive self-distances. On the one hand, this is coherent with our picture of quantum mechanics, on the other hand it can lead to some unexpected consequences. 
For example, there exist non-surjective and even non-injective quantum Wasserstein isometries (distance preserving maps) --- see \cite{GPTV23}. None of these could occur in a genuine metric setting. As a response to this phenomenon, De Palma and Trevisan introduced quantum Wasserstein divergences \cite{DPT-lecture-notes}, which are appropriately modified versions of quantum Wasserstein distances, to eliminate self-distances --- see \eqref{Wasserstein-divergence} for a precise definition. They conjectured that the divergences defined this way are genuine metrics on quantum state spaces \cite{DPT-lecture-notes}, and this conjecture has been recently justified under a certain additional assumption \cite{BPTV-PRA}.
\par
 The theory of isometries on structures of operators and functions is a well-established, active field of research in analysis and linear algebra. See, for example, \cite{ML18, IK19, MV21, Peral23}  for some of the recent advances in the field.
 In \cite{GPTV23}, the second author and his collaborators characterized the isometry group of the qubit state space with respect to the Wasserstein distance $D_{sym}$, the quantum Wasserstein distance induced by all three Pauli matrices. They obtained a Wigner-type result, showing that the isometry group consists of unitary and anti-unitary conjugations.
 They also gave lower and upper bounds on the isometry group of $D_{xz}$, the quantum Wasserstein distance induced by the Pauli operators $\sigma_x$ and $\sigma_z$.
 \par
 The complete description of the isometry group of $D_{sym}$ obtained in \cite{GPTV23} encouraged us to study the isometries of the corresponding divergence $d_{sym},$ which is more challenging for at least two reasons. First, the fact that the isometries of quantum Wasserstein distances preserve the self-distances of states provides us with a useful grab on isometries of distances --- we completely lose this tool when turning to divergences as $d(\rho,\rho)=0$ for any divergence $d$ and state $\rho.$ Second, the squared Wasserstien distances are convex, and hence the maximal distances are realized by extremal, that is, pure states. This fact is also useful when determining the structure of isometries, and we cannot use it in the case of Wasserstein divergences, as squared divergences are not convex in general.
\par
The surprising phenomenon that the structure of the $D_{xz}$-isometries turned out to be much richer and much more flexible than that of the $D_{sym}$-isometries motivated us to go one step further and remove also $\sigma_x$ from the generators of the quantum Wasserstein distance and consider $D_z$ --- we recall that we get $D_{xz}$ from $D_{sym}$ by removing $\sigma_y.$ We expected that the structure of $D_z$-isometries on qubits shows even more flexibility, and that turned out to be the case --- see Theorem \ref{thm:Dz-main} for the details.
\par
According to the above considerations, the goal of this paper is to describe isometries of the quantum Wasserstein divergence $d_{\text{sym}},$ and to characterize the isometries with respect to $D_z$, the Wasserstein distance induced by the single Pauli operator $\sigma_z$.
We describe our main results in an informal way below --- see Theorem \ref{thm:d-sym-main} and \ref{thm:Dz-main} for the precise statements.

\begin{main-result}
The isometries of the symmetric quantum Wasserstein divergence $d_{sym},$ which map pure states to pure states, are exactly the Wigner symmetries, that is, the unitary and anti-unitary conjugations. In the Bloch ball model, the group of these transformations coincides with the orthogonal group $\bO(3).$ Furthermore, a self-map of the quantum bit state space is an isometry of the quantum Wasserstein distance $D_z$ if and only if it preserves the length of the Bloch vector of every state and either keeps the third (``z") coordinate of the Bloch vectors fixed or sends them to their negative.
\end{main-result}

\subsection{Basic notions and notations} As our results concern quantum bits, we will review the basics of quantum optimal transport in the finite-dimensional setting to avoid technical difficulties that will not matter later. Let $\cH$ be a finite-dimensional complex Hilbert space, and let $\mathcal{L}(\cH)$ stand for the set of linear operators acting on $\cH$. Let $\mathcal{L}(\cH)^{\text{sa}}$ denote the real linear vector space of self-adjoint linear operators on $\cH$.
We write $A \leq B$ for $A,B \in \mathcal{L}(\cH)^{\text{sa}}$ if $B-A$ is positive semidefinite, that is, if $\langle x, (B-A)x \rangle \geq 0$ for all $x \in \cH$. We will denote by $S(\cH)$ the set of quantum states, that is $S(\cH)=\{ \rho \in \mathcal{L}(\cH) \, |\, \rho \geq 0, \,\tr_{\cH}[\rho]=1\}$. The extremal points of $S(\cH)$ are called pure states and are rank-one projections, that is, if $\rho \in S(\cH)$ is pure, then there exists a vector $\Psi \in \cH$ such that $\rho=|\Psi \rangle \langle \Psi|$. We denote the set of pure states by $\mathcal{P}_1(\cH)$. The transpose of a linear operator $A \in \cL(\cH)$ is a linear operator acting on $\cH^*,$ the (topological) dual of $\cH.$ It is denoted by $A^T$ and defined by the identity
\be \label{eq:transpose-def}
(A^T \varphi)(x) =\varphi (A x) \qquad \ler{x \in \cH, \, \varphi \in \cH^*}.
\ee
Let $\mathcal{A}=\{A_1, A_2, \dots A_N\}$ be a set of  self-adjoint linear operators on $\cH$.
Then the \textit{cost operator} $C_{\cA}$ induced by $\mathcal{A}$ is defined as
\be
C_{\cA}=\sum_{j=1}^N (A_j \otimes I^T-I \otimes A_j^T)^2.
\ee
The \textit{Wasserstein distance} corresponding to $\cA$ is denoted by $D_{\cA}(.;.)$ and defined by
\be\label{Wasserstein-definition}
D_{\cA}^2(\rho, \om)=\inf \left\{\tr_{\cH \otimes \cH^* }\left[\Pi C_{\cA}\right] \, | \,  \Pi \in \cC(\rho, \om)\right\},
\ee
where $\cC(\rho, \om)$ stands for the set of couplings of $\rho$ and $\om$. 
According to the convention introduced by De Palma and Trevisan in \cite{DPT-AHP} the set of quantum couplings of $\rho, \om \in \cS(\cH)$ is given by 
$$
\cC(\rho, \om)=\{\Pi \in \cS(\cH \otimes \cH^*) \, |\, \tr_{\cH^*}[\Pi]=\om,\,\, \tr_{\cH} [\Pi] =\rho^T\}.
$$
The \textit{Wasserstein divergence} generated by the observables belonging to $\cA$ is denoted by $d_{\cA}(.;.)$ and defined as follows:
\be\label{Wasserstein-divergence}
d_{\cA}(\rho, \om):=\ler{D_{\cA}^2(\rho, \om)-\frac{1}{2}\ler{D_{\cA}^2(\rho, \rho)+D_{\cA}^2(\om, \om)}}^{\fel}.
\ee
Let us take a closer look at the two-dimensional complex Hilbert space $\cH=\C^2.$ Its state space has a particularly convenient description: elements of $\cS(\C^2)$ can be represented with vectors from $\mathbb{R}^3$ using the Bloch representation. The Bloch vector of a state $\rho$ is defined as
$$ \mathbb{R}^3 \ni \mathbf{b}_{\rho}=\left(\tr_{\cH}[\sigma_j \rho] \right)_{j=1}^3,$$
where $\sigma_j$ refers to one of the Pauli matrices
\be \label{eq:pauli-def}
\sigma_1=\begin{bmatrix}
0 & 1 \\ 1 & 0
\end{bmatrix},\,
\sigma_2=\begin{bmatrix}
0 & -i \\ i & 0
\end{bmatrix},\,
\sigma_3=\begin{bmatrix}
1 & 0 \\ 0 & -1
\end{bmatrix}.
\ee

\section{Isometries of the $d_{sym}$ quantum Wasserstein divergence}

First, we consider the case when the transport cost is symmetric in the sense that it is generated by all three Pauli matrices. That is, $\cA=\lers{\sigma_1, \sigma_2, \sigma_3}.$ We denote the corresponding cost operator by $C_{sym},$ and it is given by
\be
C_{\text{sym}}=\sum_{j=1}^3 (\sigma_j \otimes I^T-I \otimes \sigma_j^T)^2=\begin{bmatrix}
4 & 0 & 0 & -4 \\
0 & 8 & 0 & 0 \\
0 & 0 & 8 & 0 \\
-4 & 0 & 0 & 4
\end{bmatrix}.
\ee
We denote the corresponding Wasserstein distance with $D_{\text{sym}}(.,.)$, and the corresponding Wasserstein divergence with $d_{\text{sym}}(.,.)$. The main result of this section reads as follows.

\begin{theorem} \label{thm:d-sym-main}
Let $\Phi: S(\C^2) \to S(\C^2)$ be an isometry of the quantum Wasserstein divergence corresponding to the symmetric cost operator $C_{\text{sym}}$, that is,
$$
d_{\text{sym}}(\Phi(\rho),\Phi(\omega))=d_{\text{sym}}(\rho,\omega) \qquad \ler{\rho, \omega \in S(\C^2)}
$$
such that $\Phi$ maps pure states to pure states, that is $\Phi(\cP_1(\C^2)) \subset \cP_1(\C^2)$.
Then $\Phi$ is necessarily of the form of $\Phi(\rho)=U\rho U^*$, where $U$ is either a unitary or an anti-unitary operator on $\C^2$.
Conversely, unitary and anti-unitary conjugations are $d_{sym}$-isometries.
\end{theorem}

\begin{proof}
The fact that maps of the form of $\Phi(\rho)=U\rho U^*$ are indeed isometries follows from the fact that $D_{\text{sym}}$ is invariant under unitary and anti-unitary conjugations (see Lemma 1 in \cite{GPTV23}) and from the definition of $d_{\text{sym}}(\rho, \omega).$ Indeed, $D_{\text{sym}}(U\rho U^*, U \omega U^*)=D_{\text{sym}}(\rho, \omega)$ and 
$$
d_{\text{sym}}(\rho, \omega)^2=D_{\text{sym}}(\rho, \omega)^2-\frac{1}{2}\left(D_{\text{sym}}(\rho, \rho)^2+D_{\text{sym}}(\omega, \omega)^2\right),
$$
from where we get that $d_{\text{sym}}(U\rho U^*, U\omega U^*)=d_{\text{sym}}(\rho, \omega)$.
For the converse statement, let $\rho, \om \in \pur$ be  states of the form of
\be
\rho=\fel(I+x \sigma_1+y \sigma_2+z \sigma_3)=\fel \begin{bmatrix} 1+z & x-yi \\ x+yi & 1-z
\end{bmatrix}
\ee
and 
\be
\om=\fel(I+u \sigma_1+v \sigma_2+w \sigma_3)=\fel \begin{bmatrix} 1+w & u-vi \\ u+vi & 1-w
\end{bmatrix},
\ee
where $|(x, y, z)|=1$ and $|(u, v, w)|=1$, where $|.|$ denotes the Euclidean norm on $\mathbb{R}^3$.
Then both $\rho$ and $\omega$ are pure, and the only coupling between $\rho$ and $\om$ is the tensor product (or in other words, independent) coupling $\om \otimes \rho^T$ and its cost is
\be\label{indep}
\begin{gathered}
\tr_{\cH \otimes \cH^*}\left[ (\om \otimes \rho^T)C_{\text{sym}} \right]=\tr_{\cH \otimes \cH^*}\left[(\om \otimes \rho^T)\left(6 I \otimes I^T-2 \sum_{j=1}^3 \sigma_j \otimes \sigma_j^T \right) \right] \\
=6-2 \sum_{j=1}^3 \tr_{\cH}[\om \sigma_j] \cdot \tr_{\cH^*}[\rho^T \sigma_j^T]=6-2(xu+yv+zw)=6-2 \inner{\mathbf{b}_{\rho}}{\mathbf{b}_{\om}},
\end{gathered}
\ee
where $\mathbf{b}_{\rho}=(x,y,z)$ and $\mathbf{b}_{\om}=(u,v,w)$ are the Bloch vectors of $\rho$ and $\om$, and $\inner{.}{.}$ denotes the Euclidean inner product on $\mathbb{R}^3$. In the above formula \eqref{indep}, we replaced $\C^2$ by $\cH$ for notational convenience, and this will happen later as well when indicating the Hilbert space where the trace is taken.
In particular, we can compute the distance of a pure state $\rho$ from itself:
\be\label{sindep}
\begin{gathered}
\tr_{\cH \otimes \cH^*}\left[ (\rho \otimes \rho^T)C_{\text{sym}} \right]=6-2(x^2+y^2+z^2)=6-2 \inner{\mathbf{b}_{\rho}}{\mathbf{b}_{\rho}}=6-2|\mathbf{b}_{\rho}|^2=4.
\end{gathered}
\ee

Putting \eqref{indep} and \eqref{sindep} together, we get that
\be\label{div}
\begin{gathered}
d_{\text{sym}}^2(\rho, \om)=D_{\text{sym}}^2(\rho, \om)-\frac{1}{2}(D_{\text{sym}}^2(\rho, \rho)+D_{\text{sym}}^2(\om, \om)) \\
=6-2 \inner{\mathbf{b}_{\rho}}{\mathbf{b}_{\om}}-3+|\mathbf{b}_{\rho}|^2-3+|\mathbf{b}_{\om}|^2 \\
=|\mathbf{b}_{\rho}-\mathbf{b}_{\om}|^2 \leq 4.
\end{gathered}
\ee
By \eqref{div} $d_{\text{sym}}$ coincides with the Euclidean distance of the Bloch vectors on $\mathcal{P}_1(\cH)$. Therefore the action of the $d_{\text{sym}}$ isometry $\Phi$ on $\mathcal{P}_1(\cH)$ is described by an $O \in \mathbf{O}(3)$. By Proposition $VII. 5.7.$ in \cite{Simon} for any orientation preserving $O_+ \in \mathbf{SO}(3)$ there is a $U \in \mathbf{SU}(2)$ such that the action $\rho \mapsto U\rho U^*$ is described by $O_+$ in the Bloch ball model. Similarly, for an orientation reversing $O_- \in \mathbf{O}(3)$ there is an anti-unitary  operator $V$ on $\mathbb{C}^2$ such that the action $\rho \mapsto V \rho V^*$ is described by $O_-$ in the Bloch ball model. We have seen that unitary and anti-unitary conjugations are $d_{\text{sym}}$ isometries. This together with the assumption $\Phi(\mathcal{P}_1(\cH)) \subseteq \mathcal{P}_1(\cH)$ means that we can assume without the loss of generality that $\Phi(\rho)=\rho$ for all $\rho \in \mathcal{P}_1(\cH) $.
Consider the particular pure states $\rho_j=\frac{1}{2}(I+\sigma_j)$ for $j \in \{1,\, 2,\, 3\}$. Then we have
$D_{\text{sym}}^2(\rho_j, \omega)=6-2(\mathbf{b}_{\omega})_j$, and 
$$
D_{\text{sym}}^2(\Phi(\rho_j), \Phi(\omega))=D_{\text{sym}}^2(\rho_j, \Phi(\omega))=6-2(\mathbf{b}_{\Phi(\omega)})_j.
$$
From this, we get that
$$d_{\text{sym}}^2(\rho_j,\omega)=6-2(\mathbf{b}_{\omega})_j-2-\frac{1}{2}D_{\text{sym}}^2(\omega, \omega)$$ and
$$
d_{\text{sym}}^2(\rho_j, \Phi(\omega))=6-2(\mathbf{b}_{\Phi(\omega)})_j-2-\frac{1}{2}D_{\text{sym}}^2(\Phi(\omega),\Phi(\omega)).
$$
The preserver equation $d_{\text{sym}}^2(\rho_j,\omega)=d_{\text{sym}}^2(\rho_j, \Phi(\omega))$ yields to
\be \label{eq:pres-eq-cons}
\frac{1}{2}(D_{\text{sym}}^2(\Phi(\omega),\Phi(\omega))-D_{\text{sym}}^2(\omega, \omega))=2 \ler{(\mathbf{b}_{\omega})_j-(\mathbf{b}_{\Phi(\omega)})_j}.
\ee
Now we shall take into account that the self-distance $D_{\text{sym}}^2(\rho, \rho)$ is explicitly computable, and it is given by
\be \label{eq:sym-self-dist-form}
D_{\text{sym}}^2(\rho, \rho)=4\left(1-\sqrt{1-|\mathbf{b}_{\rho}|^2}\right) \,\, \forall \rho \in \mathcal{S}(\C^2).
\ee
The computation proving \eqref{eq:sym-self-dist-form} can be found in the proof of Proposition 2. in \cite{BPTV-PRA}, however, we note that a multiplicative factor of $2$ has been lost there during the proof. Therefore, the correct formula is \eqref{eq:sym-self-dist-form}, which is consistent with \eqref{sindep}. Accordingly, we rewrite \eqref{eq:pres-eq-cons} the following way:
$$
2\left(1-\sqrt{1-|\mathbf{b}_{\Phi(\omega)}|^2} \right)
-2\left(1-\sqrt{1-|\mathbf{b}_{\omega}|^2} \right)=2 \ler{(\mathbf{b}_{\omega})_j-(\mathbf{b}_{\Phi(\omega)})_j} \,\, \forall j \in \{1,2,3 \}.
$$
Let us define 
$$
\alpha:=\left(1-\sqrt{1-|\mathbf{b}_{\Phi(\omega)}|^2} \right)-\left(1-\sqrt{1-|\mathbf{b}_{\omega}|^2} \right)
=\sqrt{1-|\mathbf{b}_{\omega}|^2}-\sqrt{1-|\mathbf{b}_{\Phi(\omega)}|^2}.
$$
Then we have the following vector equation
$$
\mathbf{b}_{\omega}=\mathbf{b}_{\Phi(\omega)}+\alpha (1,1,1).
$$
Consider another orthonormal basis in $\mathbb{C}^2$ with corresponding unitary $U$, and let $\tilde{\sigma}_j=U\sigma_jU^*$. It is easy to verify that $\tilde{\sigma}_j$ is self-adjoint, unitary and of order two, similarly to $\sigma_j$ for each $j=1, 2, 3$.
We choose this new basis in a way that the corresponding $O \in \mathbf{O}(3)$ in the Bloch vector model does not map the vector $(1,1,1)$ to itself. Next, we show that the cost operator generated by the observables $\{\tilde{\sigma}_j\}_{j=1}^{3}$ coincides with the cost generated by the observables $\{\sigma_j\}_{j=1}^3$. Indeed, 

$$
\tilde{C}_{\text{sym}}=\sum_{j=1}^3 (\tilde{\sigma}_j \otimes I^T-I \otimes \tilde{\sigma}_j^T)^2=\sum_{j=1}^3 (\tilde{\sigma}_j)^2 \otimes I^T+I \otimes (\tilde{\sigma}_j^T)^2-2 \tilde{\sigma}_j \otimes (\tilde{\sigma}_j)^T
$$
\be
\begin{gathered}
=6 I \otimes I^T-2 \sum_{j=1}^3 U \sigma_j U^* \otimes (U \sigma_j U^*)^T
\\
=6 I \otimes I^T-\left( U \otimes (U^*)^T\right) \sum_{j=1}^3 \sigma_j \otimes \sigma_j^T \left( U^* \otimes U^T \right) \\
=\left( U \otimes (U^*)^T\right)C_{\text{sym}} \left( U^* \otimes U^T \right).
\end{gathered}
\ee
According to Lemma 1 in \cite{GPTV23}, $C_{\text{sym}}$ is invariant under unitary conjugations, that is, $\left( U \otimes (U^*)^T\right)C_{\text{sym}} \left( U^* \otimes U^T \right)=C_{\text{sym}}.$
We consider the distinguished pure states $\Tilde{\rho}_j := \frac{1}{2}\left(I+\Tilde{\sigma}_j \right).$
Then $D_C^2(\Tilde{\rho}_j,\omega)=6-2 \langle \Tilde{\mathbf{b}}_{\omega}, e_j \rangle=\Tilde{\mathbf{b}}_{{\omega}_j}$, where $\Tilde{\mathbf{b}}_{\omega}$ denotes the Bloch vector of $\omega$ in the new basis, that is, $\Tilde{\mathbf{b}}_{\omega}=\ler{\tr_{\cH}\lesq{\tilde{\sigma}_j \omega}}_{j=1}^3.$ The latter equality holds because the symmetric cost operator is invariant under unitary conjugations, see Lemma 1 in \cite{GPTV23}.
Thus the preserver equation $d_{\text{sym}}^2(\Tilde{\rho}_j, \omega)=d_{\text{sym}}^2(\Tilde{\rho}_j, \Phi(\omega))$ yields to
$$\Tilde{\mathbf{b}}_{\omega}=\Tilde{\mathbf{b}}_{\Phi(\omega)}+\alpha (1,1,1).$$
Let $  (1,1,1) \neq \mathbf{v} \in \mathbb{R}^3
$ be the image of $(1,1,1)$ under $O$. Then applying $O$ to  both sides of the equation
$$\mathbf{b}_{\omega}=\mathbf{b}_{\Phi(\omega)}+\alpha (1,1,1)$$
yields to 
$$
0=\alpha \cdot \left((1,1,1)-\mathbf{v} \right).
$$
Since $\mathbf{v} \neq (1,1,1),$ we conclude that $\alpha=0$. This means that $\Phi(\rho)=\rho$ for all $\rho \in \cS(\cH),$ which concludes the proof.
\end{proof}

\section{Isometries of the $D_z$ quantum Wasserstein distance}
In this section, we consider the cost operator $C_z$ generated by the single observable $\sigma_z:=\sigma_3=\begin{bmatrix}
    1 & 0 \\ 0 & -1
\end{bmatrix},$ that is,
\be \label{eq:Cz-def}
C_z=(\sigma_z \otimes I^T-I \otimes \sigma_z^T)^2=\begin{bmatrix} 0 & 0 & 0 & 0 \\ 0 & 4 & 0 & 0 \\ 0 & 0 & 4 & 0 \\ 0 & 0 & 0 & 0 \end{bmatrix}.
\ee
Then the corresponding Wasserstein distance is given by 
$$D_z^2(\rho, \omega)=\inf\{\tr_{\cH \otimes \cH^*}[C_z \Pi]\,| \, \Pi \in C(\rho, \omega)\}.$$
The main result of this section is the complete classification of the isometries of the qubit state space which reads as follows.

\begin{theorem} \label{thm:Dz-main}
Let $D_z$ denote the quantum Wasserstein distance defined by the cost operator $C_z$ introduced in \eqref{eq:Cz-def}, and let $\Phi: \cS(\C^2) \rightarrow\cS(\C^2)$ be a map. Then the following are equivalent.
\begin{enumerate}
    \item \label{it:isom} The map $\Phi$ is a quantum Wasserstein isometry with respect to $D_z,$ that is, $D_z(\Phi(\rho),\Phi(\omega))=D_z(\rho,\omega)$ for all $\rho, \omega \in \cS(\C^2).$
    \item \label{it:bloch-char} The map $\Phi$ leaves the Euclidean length of the Bloch vector of a state invariant, that is, $\abs{\bb_{\Phi(\rho)}}=\abs{\bb_{\rho}}$ for all $\rho \in \cS(\C^2),$ and either $\ler{\bb_{\Phi(\rho)}}_3=\ler{\bb_{\rho}}_3$ for all $\rho \in \cS(\C^2),$ or $\ler{\bb_{\Phi(\rho)}}_3=-\ler{\bb_{\rho}}_3$ for all $\rho \in \cS(\C^2).$
\end{enumerate}

\end{theorem}

\begin{proof}
Let us consider the direction \eqref{it:isom} $\Longrightarrow$ \eqref{it:bloch-char} first.
\par
{\bf Step 1}: the $D_z$-diameter of $\cS(\C^2)$ is $2$ and it is realized if and only if 
$$
\lers{\rho,\omega}=\lers{\fel(I+\sigma_z),\fel(I-\sigma_z)}.
$$
Indeed, let $\rho$ and $\om$ be arbitrary states, then the cost of the trivial coupling $\om \otimes \rho^T$ is an upper bound of $D_z^2(\rho, \om)$, that is,
\begin{equation}
D_z^2(\rho, \om) \leq \tr_{\cH \otimes \cH^*}\left[\om \otimes \rho^T \cdot (2 I \otimes I^T-2 \sigma_z \otimes \sigma_z^T)\right]=2-2 (b_{\rho})_3 \cdot (b_{\om})_3.
\end{equation}
It follows that $\text{diam}(S(\cH), D_z^2) \leq 4$ and $D_z^2(\rho, \omega)=\text{diam}(S(\cH), D_z^2)$ if and only if $\rho=\frac{1}{2}(I + \sigma_z)$ and $\om=\frac{1}{2}(I-\sigma_z),$ or the other way around, since if $|(\mathbf{b}_{\rho})_3|<1$ or $|(\mathbf{b}_{\om})_3|<1$, then $2-2 (b_{\rho})_3 \cdot (b_{\om})_3<4$.
Therefore, $\Phi\ler{\fel(I \pm \sigma_z)}=\fel(I\pm \sigma_z)$ or $\Phi\ler{\fel(I \pm \sigma_z)}=\fel(I\mp \sigma_z),$ that is, $\Phi$ either leaves both $\fel(I+ \sigma_z)$ and $\fel(I-\sigma_z)$ fixed, or it switches them.
Consider the map $T:S(\cH) \to S(\cH)$ where $T(\rho)=\sigma_x \rho \sigma_x^*$ for all $\rho \in \cS(\C^2)$. 
Then $T$ is an isometry of $\cS(\C^2)$ with respect to the quantum Wasserstein distance $D_z$ such that $T\left(\frac{1}{2}(I + \sigma_z)\right)=\frac{1}{2}(I - \sigma_z)$ and $T\left(\frac{1}{2}(I - \sigma_z)\right)=\frac{1}{2}(I + \sigma_z)$.
Indeed,
\be
T\left(\frac{1}{2}(I + \sigma_z)\right)=\sigma_x\left(\frac{1}{2}(I+\sigma_z)\right)\sigma_x^*=\frac{1}{2}(I+\sigma_x \sigma_z \sigma_x^*)=\frac{1}{2}(I-\sigma_z).
\ee
Similarly, $T\left(\frac{1}{2}(I - \sigma_z)\right)=\frac{1}{2}(I + \sigma_z)$.
To see that $T$ is isometric, it is sufficient to show that for all $\rho,\, \om$ in $\cS(\C^2)$ and for all couplings $\pi$ in $\cC(\rho, \om)$ there  exists a coupling $\tilde{\pi}$ in $\cC(T(\rho), T(\om))$ such that $\tr_{\hohc}[\tilde{\pi}C_z]=\tr_{\hohc}[\pi C_z],$ and the other way around: for all $\tilde{\pi}$ in $\cC(T(\rho), T(\om))$ there is a coupling $\pi \in \cC(\rho,\om)$ with the same cost. The desired bijection between $\cC(\rho, \om)$ and $\cC(T(\rho),T(\omega))$ reads as follows: for a given $\pi \in \cC(\rho, \omega)$ we define $\tilde{\pi}$ by $\tilde{\pi}=(\sigma_x \otimes \sigma_x^T)\pi (\sigma_x \otimes \sigma_x^T).$ Now
\be
\begin{gathered}
 \tr_{\hohc}[\tilde{\pi} C_z]=\tr_{\hohc}[(\sigma_x \otimes \sigma_x^T)\pi (\sigma_x \otimes \sigma_x^T)(\sigma_z \otimes I^T- I\otimes \sigma_z^T)^2)] \\
 =\tr_{\hohc}[\pi(\sigma_x \otimes \sigma_x^T)(2I-2\sigma_z \otimes \sigma_z^T)(\sigma_x \otimes \sigma_x^T)] \\
 =\tr_{\hohc}[\pi(2I-2 \sigma_x \sigma_z \sigma_x \otimes (\sigma_x \sigma_z \sigma_x)^T)] \\
 =\tr_{\hohc}[\pi(2I-2 (-\sigma_z)\otimes (-\sigma_z)^T)]=\tr_{\hohc}[\pi C_z].
\end{gathered}
\ee
Therefore, we can assume, without loss of generality, that if $\Phi$ is a $D_z-$Wasserstein  isometry, then $\Phi\ler{\frac{1}{2}(I+\sigma_z)}=\frac{1}{2}(I+\sigma_z)$ and $\Phi\ler{\frac{1}{2}(I-\sigma_z)}=\frac{1}{2}(I-\sigma_z)$.

\par 
{\bf Step 2}: Consequently, $\ler{\bb_{\Phi(\rho)}}_3= \ler{\bb_{\rho}}_3$ as 
$$
D_z\ler{\Phi(\rho), \Phi\ler{\fel(I+ \sigma_z)}}=D_z\ler{\Phi(\rho), \fel(I+ \sigma_z)}=D_z\ler{\rho, \fel(I + \sigma_z)}.
$$
From this, we have that
\be
2-2\ler{\bb_{\Phi(\rho)}}_3 \cdot  1=2-2\ler{\bb_{\rho}}_3 \cdot 1.
\ee
It follows that $\ler{\bb_{\Phi(\rho)}}_3 = \ler{\bb_{\rho}}_3.$

\par
{\bf Step 3}:
Next we utilise that $\Phi$ preserves the self-distances, that is, $D_z(\Phi(\rho), \Phi(\rho))=D_z(\rho, \rho)$ for all $\rho$ in $\cS(\cH).$
According to Corollary 1. in \cite{DPT-AHP} the self-distance is realised by the canonical purification of the state:
\be \label{eq:self-dist-form}
D^2_z(\rho, \rho)=\langle \langle \sqrt{\rho}|| C_z ||\sqrt{\rho} \rangle \rangle.
\ee
Here, and in the sequel, we use the following notational convention originally introduced in \cite{DPT-AHP}. For an operator $X \in \cL(\C^2),$  the symbol $||X\rangle\rangle$ stands for the linear map 
$$
\C \ni z \mapsto z e(X) \in \C^2 \otimes (\C^2)^*
$$
where $e: \cL(\C^2) \to \C^2 \otimes (\C^2)^*$ is the canonical linear isometry between $\cL(\C^2)$ and $\C^2 \otimes (\C^2)^*.$ The symbol $\langle\langle X ||$ denotes the linear functional 
$$
\C^2 \otimes (\C^2)^* \ni Y \mapsto \tr_{\C^2}\left[X^* e^{-1}(Y)\right] \in \C.
$$
A spectral decomposition of $C_z$ is 
$$
C_z=
4\ler{\fel\lesq{\ba{cccc} 0&0&0&0 \\ 0&1&1&0 \\ 0&1&1&0 \\ 0&0&0&0 \ea}+\fel\lesq{\ba{cccc} 0&0&0&0 \\ 0&1&-1&0 \\ 0&-1&1&0 \\ 0&0&0&0 \ea} }
$$
$$
=4\ler{\fel||\sigma_1\rangle \rangle \langle \langle \sigma_1||+\fel||\sigma_2 \rangle \rangle \langle \langle \sigma_2||},
$$
and if $\rho \neq \fel I$ then
$$\sqrt{\rho}=\sqrt{\lambda} \fel \ler{I+\frac{\bb_{\rho}}{|\bb_{\rho}|} \cdot \underline{\sigma}}+\sqrt{1-\lambda}\fel\ler{I-\frac{\bb_{\rho}}{|\bb_{\rho}|} \cdot \underline{\sigma}},$$
where $\lambda=\fel(1+|\mathbf{b}_{\rho}|)$ and $\underline{\sigma}$ is the vector formed by the Pauli matrices defined in \eqref{eq:pauli-def}.
Therefore,
\be \label{eq:pur-coupl-cost}
\begin{gathered}
\langle \langle \sqrt{\rho}|| C_z ||\sqrt{\rho} \rangle \rangle= \\
=\ler{\sqrt{\lambda} \fel \ler{\langle \langle \sigma_0||+\sum\limits_{j=1}^3 \frac{(\mathbf{b}_{\rho})_j}{|\mathbf{b}_{\rho}|} \langle \langle \sigma_j||}+\sqrt{1-\lambda} \fel\ler{\langle \langle \sigma_0||-\sum\limits_{j=1}^3 \frac{(\mathbf{b}_{\rho})_j}{|\mathbf{b}_{\rho}|} \langle \langle \sigma_j||}} \cdot \\
\cdot  \ler{2 ||\sigma_1 \rangle \rangle \langle \langle \sigma_1||+2 ||\sigma_2 \rangle \rangle \langle \langle \sigma_2||} \cdot 
\\
\cdot \ler{\sqrt{\lambda} \fel \ler{ ||\sigma_0 \rangle \rangle+\sum\limits_{j=1}^3 \frac{(\mathbf{b}_{\rho})_j}{|\mathbf{b}_{\rho}|} ||\sigma_j \rangle \rangle }+\sqrt{1-\lambda} \fel\ler{ ||\sigma_0 \rangle \rangle -\sum\limits_{j=1}^3 \frac{(\mathbf{b}_{\rho})_j}{|\mathbf{b}_{\rho}|} || \sigma_j \rangle \rangle}}.
\end{gathered}
\ee
Note that $\langle \langle \sigma_i|| \sigma_j \rangle \rangle=\tr_{\cH}[\sigma_i^* \sigma_j]=2 \delta_{ij}$. Therefore, the above formula \eqref{eq:pur-coupl-cost} greatly simplifies, and by \eqref{eq:self-dist-form} we get that
$$
D^2_z(\rho, \rho)=\fel \ler{1-2 \sqrt{\lambda (1-\lambda )}} \frac{(\mathbf{b}_{\rho})_1^2+(\mathbf{b}_{\rho})_2^2}{|\mathbf{b}_{\rho}|^2}
$$
\be 
=\fel \ler{1-\sqrt{1-|\mathbf{b}_{\rho}|^2}} \cdot \ler{1-\frac{(\mathbf{b}_{\rho})^2_3}{|\mathbf{b}_{\rho}|^2}}.
\ee
From the preserver equation $D_z^2(\Phi(\rho), \Phi(\rho))=D_z^2(\rho, \rho)$ we get that
$$
\fel \ler{1-\sqrt{1-|\mathbf{b}_{\Phi(\rho)}|^2}} \cdot \ler{1-\frac{(\mathbf{b}_{\Phi(\rho}))^2_3}{|\mathbf{b}_{\Phi(\rho)}|^2}}
=
\fel \ler{1-\sqrt{1-|\mathbf{b}_{\rho}|^2}} \cdot \ler{1-\frac{(\mathbf{b}_{\rho})^2_3}{|\mathbf{b}_{\rho}|^2}}.
$$
Note that we already deduced in Step 2 that $(\mathbf{b}_{\Phi(\rho)})^2_3=(\mathbf{b}_{\rho})^2_3,$ and the map $t \mapsto \ler{1-\sqrt{1-t}} \cdot \ler{1-\frac{c}{t}}$ is strictly monotone increasing on the domain $t \in [0, 1]$ for any fixed $c \in [0,1].$ Therefore, $|\mathbf{b}_{\Phi(\rho)}|=|\mathbf{b}_{\rho}|$ which completes the proof of the direction \eqref{it:isom} $\Longrightarrow$ \eqref{it:bloch-char}.
\par
Now we consider the direction \eqref{it:bloch-char} $\Longrightarrow$ \eqref{it:isom}.
\par
Consider the case $\ler{\bb_{\Phi(\rho)}}_3 = \ler{\bb_{\rho}}_3$ for all $\rho \in \cS(\C^2).$ This equation and the invariance of the length of the Bloch vector $\abs{\bb_{\Phi(\rho)}}=\abs{\bb_{\rho}}$ implies that for each $\rho$ there is some $t(\rho) \in \R$ (which depends on $\rho$ !) such that
$$
\Phi(\rho)=U_z(t(\rho)) \rho U_z(t(\rho))^*
$$
where $U_z(t)=e^{it\sigma_z}.$
\par
Now the couplings of $\Phi(\rho)$ and $\Phi(\omega)$ take the following form:
\be\label{couplings}
\cC\ler{\Phi(\rho),\Phi(\omega)}=\lers{\cU \Pi \cU^* \, \middle| \, \Pi \in \cC(\rho, \omega)}
\ee
where $\cU=U_z(t(\omega))\otimes \ler{U_z(t(\rho))^*}^T.$
Indeed, let $\Pi=\sum\limits_{k=1}^K A_k \otimes B_k^T$ be a decomposition of a coupling $\Pi \in  \cC(\rho, \omega)).$
Then 
\be
\begin{split}
\cU \Pi \cU^*&=\sum\limits_{k=1}^K U_z(t (\om)) \otimes \ler{U_z(t(\rho))^*}^T \ler{A_k \otimes B_k^T} U_z(t (\om))^* \otimes \ler{U_z(t(\rho))}^T \\
&=\sum\limits_{k=1}^K U_z(t (\om)) A_k U_z(t (\om))^* \otimes \ler{U_z(t(\rho))B_k U_z(t(\rho))^*}^T.
\end{split}
\ee
We have that 
\be \label{eq:q-marg-comp}
\begin{split}
\tr_{\cH^*}[\cU \Pi \cU^*]&=\sum\limits_{k=1}^K  U_z(t (\om)) A_k  U_z(t (\om))^* \cdot \tr_{\cH} [B_k] \\
&=U_z(t (\om))  \ler{\sum\limits_{k=1}^K \tr_{\cH} [B_k] \cdot A_k } U_z(t (\om))^* \\
&=U_z(t (\om)) \tr_{\cH^*}[\Pi]U_z(t (\om))^*=U_z(t (\om)) \om U_z(t (\om))^* \\
&= \Phi(\om).
\end{split}
\ee
Similarly, $\tr_{\cH}[\cU \Pi \cU^*]=\Phi(\rho)^T$. This shows that $\lers{\cU \Pi \cU^* \, \middle| \, \Pi \in \cC(\rho, \omega)} \subseteq \cC\ler{\Phi(\rho),\Phi(\omega)}$. For the reversed inclusion, a computation very similar to \eqref{eq:q-marg-comp} shows that if $\tilde{\Pi}$ is a coupling of $\Phi(\rho)$ and $\Phi(\omega)$, then $\cU^* \tilde{\Pi} \cU$ belongs to $\cC(\rho, \om)$.
It is a crucial observation that $\cU$ commutes with $C_z$ for all $t(\rho), t(\om) \in \mathbb{R}$. Indeed, $\cU=e^{it(\om)\sigma_z} \otimes \ler{e^{it(\rho)\sigma_z}}^T$ and $C_z=(\sigma_z \otimes I^T -I \otimes \sigma_z^T)^2.$ All terms can be obtained from the continuous function calculus of $\sigma_z$, therefore $[\cU, C_z ]=0$.
Finally, by \eqref{couplings} we get that
\be
\begin{split}
D^2_z(\Phi(\rho), \Phi(\om))&=\inf\{\tr_{\cH \otimes \cH^*}[\tilde{\Pi}C_z] \,| \, \tilde{\Pi} \in \cC(\Phi(\rho), \Phi(\om))\}= \\
&=\inf\{\tr_{\cH \otimes \cH^*}[\cU \Pi \cU^* C_z]\,| \, \Pi \in \cC(\rho, \om)\}= \\
&= \inf\{\tr_{\cH \otimes \cH^*}[\Pi \cU^*\cU  C_z]\,| \, \Pi \in \cC(\rho, \om)\}= \\
&=D^2_z(\rho, \om).
\end{split}
\ee
which completes the proof of the direction \eqref{it:bloch-char} $\Longrightarrow$ \eqref{it:isom}.
\end{proof}

\paragraph*{{\bf Final remarks.}}
Given the existing results on the isometries of the quantum bit state space with respect to the quantum Wasserstein distances $D_{sym}$ and $D_{xz}$ (see \cite{GPTV23}), and also our present results on the preservers of the divergence $d_{sym}$ and the distance $D_z,$ it is highly natural to ask what can be said about the structure of the qubit state space isometries with respect to the quantum Wasserstein divergences $d_{xz}$ and $d_z.$ We believe that these questions are highly nontrivial --- especially as quantum Wasserstein \emph{divergences} are typically non-convex (in contract to quantum Wasserstein \emph{distances}), and hence it is often the case that their extremal values (maxima) are not realaized by extremal (that is: pure) states ---, and we consider them as attractive open problems.
\par
\paragraph*{{\bf Acknowledgment.}} We are grateful to the anonymous reviewer for his/her valuable comments and suggestions, and insightful questions.


\bibliographystyle{plainurl}  
\bibliography{qw-isom-qubit}

\begin{thebibliography}{10}

\bibitem{Ambrosio-lecturenotes}
L.~Ambrosio, L.~A. Caffarelli, Y.~Brenier, G.~Buttazzo, and C.~Villani.
\newblock {\em Optimal transportation and applications}, volume 1813 of {\em Lecture Notes in Mathematics}.
\newblock Springer-Verlag, Berlin; Centro Internazionale Matematico Estivo (C.I.M.E.), Florence, 2003.
\newblock Lectures from the C.I.M.E. Summer School held in Martina Franca, September 2--8, 2001.
\newblock \href {https://doi.org/10.1007/b12016} {\path{doi:10.1007/b12016}}.

\bibitem{BP}
J.~Backhoff-Veraguas and G.~Pammer.
\newblock Stability of martingale optimal transport and weak optimal transport.
\newblock {\em Ann. Appl. Probab.}, 32(1):721--752, 2022.

\bibitem{BiegelbockJuillet}
Mathias Beiglb\"{o}ck and Nicolas Juillet.
\newblock On a problem of optimal transport under marginal martingale constraints.
\newblock {\em Ann. Probab.}, 44(1):42--106, 2016.
\newblock \href {https://doi.org/10.1214/14-AOP966} {\path{doi:10.1214/14-AOP966}}.

\bibitem{BianeVoiculescu}
Philippe Biane and Dan Voiculescu.
\newblock A free probability analogue of the {W}asserstein metric on the trace-state space.
\newblock {\em Geom. Funct. Anal.}, 11(6):1125--1138, 2001.
\newblock \href {https://doi.org/10.1007/s00039-001-8226-4} {\path{doi:10.1007/s00039-001-8226-4}}.

\bibitem{BistronEcksteinZyczkowski}
Rafał Bistroń, Michał Eckstein, and Karol Życzkowski.
\newblock Monotonicity of the quantum 2-{W}asserstein distance, 2022.
\newblock \href {https://arxiv.org/abs/2204.07405} {\path{arXiv:2204.07405}}.

\bibitem{BPTV-PRA}
Gergely Bunth, J\'ozsef Pitrik, Tam\'as Titkos, and D\'aniel Virosztek.
\newblock Metric property of quantum wasserstein divergences.
\newblock {\em Phys. Rev. A}, 110:022211, Aug 2024.
\newblock URL: \url{https://link.aps.org/doi/10.1103/PhysRevA.110.022211}, \href {https://doi.org/10.1103/PhysRevA.110.022211} {\path{doi:10.1103/PhysRevA.110.022211}}.

\bibitem{CagliotiGolsePaul}
Emanuele Caglioti, Fran\c{c}ois Golse, and Thierry Paul.
\newblock Quantum optimal transport is cheaper.
\newblock {\em J. Stat. Phys.}, 181(1):149--162, 2020.
\newblock \href {https://doi.org/10.1007/s10955-020-02571-7} {\path{doi:10.1007/s10955-020-02571-7}}.

\bibitem{CagliotiGolsePaul-towardsqot}
Emanuele Caglioti, François Golse, and Thierry Paul.
\newblock Towards optimal transport for quantum densities, 2021.
\newblock \href {https://arxiv.org/abs/2101.03256} {\path{arXiv:2101.03256}}.

\bibitem{CarlenMaas-3}
Eric~A. Carlen and Jan Maas.
\newblock Gradient flow and entropy inequalities for quantum {M}arkov semigroups with detailed balance.
\newblock {\em J. Funct. Anal.}, 273(5):1810--1869, 2017.
\newblock \href {https://doi.org/10.1016/j.jfa.2017.05.003} {\path{doi:10.1016/j.jfa.2017.05.003}}.

\bibitem{CarlenMaas-1}
Eric~A. Carlen and Jan Maas.
\newblock Correction to: {N}on-commutative calculus, optimal transport and functional inequalities in dissipative quantum systems.
\newblock {\em J. Stat. Phys.}, 181(6):2432--2433, 2020.
\newblock \href {https://doi.org/10.1007/s10955-020-02671-4} {\path{doi:10.1007/s10955-020-02671-4}}.

\bibitem{CarlenMaas-2}
Eric~A. Carlen and Jan Maas.
\newblock Non-commutative calculus, optimal transport and functional inequalities in dissipative quantum systems.
\newblock {\em J. Stat. Phys.}, 178(2):319--378, 2020.
\newblock \href {https://doi.org/10.1007/s10955-019-02434-w} {\path{doi:10.1007/s10955-019-02434-w}}.

\bibitem{DattaRouze2}
Nilanjana Datta and Cambyse Rouz\'{e}.
\newblock Concentration of quantum states from quantum functional and transportation cost inequalities.
\newblock {\em J. Math. Phys.}, 60(1):012202, 22, 2019.
\newblock \href {https://doi.org/10.1063/1.5023210} {\path{doi:10.1063/1.5023210}}.

\bibitem{DattaRouze1}
Nilanjana Datta and Cambyse Rouz\'{e}.
\newblock Relating relative entropy, optimal transport and {F}isher information: a quantum {HWI} inequality.
\newblock {\em Ann. Henri Poincar\'{e}}, 21(7):2115--2150, 2020.
\newblock \href {https://doi.org/10.1007/s00023-020-00891-8} {\path{doi:10.1007/s00023-020-00891-8}}.

\bibitem{DL-DP-M-T}
Giuseppe~Bruno De~Luca, Nicol\`o De~Ponti, Andrea Mondino, and Alessandro Tomasiello.
\newblock Gravity from thermodynamics: optimal transport and negative effective dimensions.
\newblock {\em SciPost Phys.}, 15(2):Paper No. 039, 55, 2023.
\newblock \href {https://doi.org/10.21468/scipostphys.15.2.039} {\path{doi:10.21468/scipostphys.15.2.039}}.

\bibitem{DePalmaMarvianTrevisanLloyd}
Giacomo De~Palma, Milad Marvian, Dario Trevisan, and Seth Lloyd.
\newblock The quantum {W}asserstein distance of order 1.
\newblock {\em IEEE Trans. Inform. Theory}, 67(10):6627--6643, 2021.
\newblock \href {https://doi.org/10.1109/TIT.2021.3076442} {\path{doi:10.1109/TIT.2021.3076442}}.

\bibitem{DPT-AHP}
Giacomo De~Palma and Dario Trevisan.
\newblock Quantum optimal transport with quantum channels.
\newblock {\em Ann. Henri Poincar\'{e}}, 22(10):3199--3234, 2021.
\newblock \href {https://doi.org/10.1007/s00023-021-01042-3} {\path{doi:10.1007/s00023-021-01042-3}}.

\bibitem{DPT-lecture-notes}
Giacomo De~Palma and Dario Trevisan.
\newblock Quantum optimal transport: quantum channels and qubits.
\newblock In Jan Maas, Simone Rademacher, Tam\'as Titkos, and D\'aniel Virosztek, editors, {\em Optimal Transport on Quantum Structures}. Springer Nature, 2024.
\newblock \href {https://arxiv.org/abs/2307.16268} {\path{arXiv:2307.16268}}.

\bibitem{Duvenhage1}
Rocco Duvenhage.
\newblock Optimal quantum channels.
\newblock {\em Phys. Rev. A}, 104(3):Paper No. 032604, 8, 2021.
\newblock \href {https://doi.org/10.1103/physreva.104.032604} {\path{doi:10.1103/physreva.104.032604}}.

\bibitem{Duvenhage2}
Rocco Duvenhage, Samuel Skosana, and Machiel Snyman.
\newblock Extending quantum detailed balance through optimal transport, 2022.
\newblock \href {https://arxiv.org/abs/2206.15287} {\path{arXiv:2206.15287}}.

\bibitem{Peyre1}
Sira Ferradans, Nicolas Papadakis, Gabriel Peyr\'{e}, and Jean-Fran\c{c}ois Aujol.
\newblock Regularized discrete optimal transport.
\newblock {\em SIAM J. Imaging Sci.}, 7(3):1853--1882, 2014.
\newblock \href {https://doi.org/10.1137/130929886} {\path{doi:10.1137/130929886}}.

\bibitem{Figalli-note}
Alessio Figalli.
\newblock {\em The {M}onge-{A}mp\`ere equation and its applications}.
\newblock Zurich Lectures in Advanced Mathematics. European Mathematical Society (EMS), Z\"{u}rich, 2017.
\newblock \href {https://doi.org/10.4171/170} {\path{doi:10.4171/170}}.

\bibitem{FriedmanEcksteinColeZyczkowski-MK}
Shmuel Friedland, Micha\l Eckstein, Sam Cole, and Karol \.{Z}yczkowski.
\newblock Quantum {M}onge-{K}antorovich {P}roblem and {T}ransport {D}istance between {D}ensity {M}atrices.
\newblock {\em Phys. Rev. Lett.}, 129(11):Paper No. 110402, 2022.
\newblock \href {https://doi.org/10.1103/physrevlett.129.110402} {\path{doi:10.1103/physrevlett.129.110402}}.

\bibitem{GPTV23}
Gy{\"o}rgy~P{\'a}l Geh{\'e}r, J{\'o}zsef Pitrik, Tam{\'a}s Titkos, and D{\'a}niel Virosztek.
\newblock Quantum {W}asserstein isometries on the qubit state space.
\newblock {\em J. Math. Anal. Appl.}, 522:126955, 2023.

\bibitem{GolseMouhotPaul}
Fran\c{c}ois Golse, Cl\'{e}ment Mouhot, and Thierry Paul.
\newblock On the mean field and classical limits of quantum mechanics.
\newblock {\em Comm. Math. Phys.}, 343(1):165--205, 2016.
\newblock \href {https://doi.org/10.1007/s00220-015-2485-7} {\path{doi:10.1007/s00220-015-2485-7}}.

\bibitem{GolsePaul-Schrodinger}
Fran\c{c}ois Golse and Thierry Paul.
\newblock The {S}chr\"{o}dinger equation in the mean-field and semiclassical regime.
\newblock {\em Arch. Ration. Mech. Anal.}, 223(1):57--94, 2017.
\newblock \href {https://doi.org/10.1007/s00205-016-1031-x} {\path{doi:10.1007/s00205-016-1031-x}}.

\bibitem{GolsePaul-wavepackets}
Fran\c{c}ois Golse and Thierry Paul.
\newblock Wave packets and the quadratic {M}onge-{K}antorovich distance in quantum mechanics.
\newblock {\em C. R. Math. Acad. Sci. Paris}, 356(2):177--197, 2018.
\newblock \href {https://doi.org/10.1016/j.crma.2017.12.007} {\path{doi:10.1016/j.crma.2017.12.007}}.

\bibitem{GolsePaul-meanfieldlimit}
Fran\c{c}ois Golse and Thierry Paul.
\newblock Empirical measures and quantum mechanics: applications to the mean-field limit.
\newblock {\em Comm. Math. Phys.}, 369(3):1021--1053, 2019.
\newblock \href {https://doi.org/10.1007/s00220-019-03357-z} {\path{doi:10.1007/s00220-019-03357-z}}.

\bibitem{HMS}
M.~Hairer, J.~C. Mattingly, and M.~Scheutzow.
\newblock Asymptotic coupling and a general form of {H}arris' theorem with applications to stochastic delay equations.
\newblock {\em Probab. Theory Related Fields}, 149(1-2):223--259, 2011.

\bibitem{MV21}
Maliheh Hosseini and A.~Jim\'enez-Vargas.
\newblock Approximate local isometries on spaces of absolutely continuous functions.
\newblock {\em Results Math.}, 76(2):Paper No. 72, 16, 2021.
\newblock \href {https://doi.org/10.1007/s00025-021-01384-8} {\path{doi:10.1007/s00025-021-01384-8}}.

\bibitem{IK19}
Dijana Ili\v~sevi\'c and Bojan Kuzma.
\newblock On square roots of isometries.
\newblock {\em Linear Multilinear Algebra}, 67(9):1898--1921, 2019.
\newblock \href {https://doi.org/10.1080/03081087.2018.1474847} {\path{doi:10.1080/03081087.2018.1474847}}.

\bibitem{JKO-98}
Richard Jordan, David Kinderlehrer, and Felix Otto.
\newblock The variational formulation of the {F}okker-{P}lanck equation.
\newblock {\em SIAM J. Math. Anal.}, 29(1):1--17, 1998.
\newblock \href {https://doi.org/10.1137/S0036141096303359} {\path{doi:10.1137/S0036141096303359}}.

\bibitem{Lott}
John Lott.
\newblock Some geometric calculations on {W}asserstein space.
\newblock {\em Comm. Math. Phys.}, 277(2):423--437, 2008.
\newblock \href {https://doi.org/10.1007/s00220-007-0367-3} {\path{doi:10.1007/s00220-007-0367-3}}.

\bibitem{ML18}
Lajos Moln\'ar.
\newblock Busch-{G}udder metric on the cone of positive semidefinite operators and its isometries.
\newblock {\em Integral Equations Operator Theory}, 90(2):Paper No. 20, 20, 2018.
\newblock \href {https://doi.org/10.1007/s00020-018-2443-9} {\path{doi:10.1007/s00020-018-2443-9}}.

\bibitem{OPS}
Carlo Orrieri, Alessio Porretta, and Giuseppe Savar\'{e}.
\newblock A variational approach to the mean field planning problem.
\newblock {\em J. Funct. Anal.}, 277(6):1868--1957, 2019.
\newblock \href {https://doi.org/10.1016/j.jfa.2019.04.011} {\path{doi:10.1016/j.jfa.2019.04.011}}.

\bibitem{Peral23}
Antonio~M. Peralta.
\newblock Preservers of triple transition pseudo-probabilities in connection with orthogonality preservers and surjective isometries.
\newblock {\em Results Math.}, 78(2):Paper No. 51, 23, 2023.
\newblock \href {https://doi.org/10.1007/s00025-022-01827-w} {\path{doi:10.1007/s00025-022-01827-w}}.

\bibitem{Santambrogio}
Filippo Santambrogio.
\newblock {\em Optimal transport for applied mathematicians}, volume~87 of {\em Progress in Nonlinear Differential Equations and their Applications}.
\newblock Birkh\"{a}user/Springer, Cham, 2015.
\newblock Calculus of variations, PDEs, and modeling.
\newblock \href {https://doi.org/10.1007/978-3-319-20828-2} {\path{doi:10.1007/978-3-319-20828-2}}.

\bibitem{Peyre2}
Morgan~A. Schmitz, Matthieu Heitz, Nicolas Bonneel, Fred Ngol\`e, David Coeurjolly, Marco Cuturi, Gabriel Peyr\'{e}, and Jean-Luc Starck.
\newblock Wasserstein dictionary learning: optimal transport-based unsupervised nonlinear dictionary learning.
\newblock {\em SIAM J. Imaging Sci.}, 11(1):643--678, 2018.
\newblock \href {https://doi.org/10.1137/17M1140431} {\path{doi:10.1137/17M1140431}}.

\bibitem{Simon}
Barry Simon.
\newblock {\em Representations of finite and compact groups}, volume~10 of {\em Graduate Studies in Mathematics}.
\newblock American Mathematical Society, Providence, RI, 1996.
\newblock \href {https://doi.org/10.1038/383266a0} {\path{doi:10.1038/383266a0}}.

\bibitem{TothPitrik}
G{\'{e}}za T{\'{o}}th and J{\'{o}}zsef Pitrik.
\newblock Quantum {W}asserstein distance based on an optimization over separable states.
\newblock {\em {Quantum}}, 7:1143, October 2023.
\newblock \href {https://doi.org/10.22331/q-2023-10-16-1143} {\path{doi:10.22331/q-2023-10-16-1143}}.

\bibitem{ZyczkowskiSlomczynski1}
Karol \.{Z}yczkowski and Wojciech S{\l}omczy\'{n}ski.
\newblock The {M}onge metric on the sphere and geometry of quantum states.
\newblock {\em J. Phys. A}, 34(34):6689--6722, 2001.
\newblock \href {https://doi.org/10.1088/0305-4470/34/34/311} {\path{doi:10.1088/0305-4470/34/34/311}}.

\bibitem{ZyczkowskiSlomczynski2}
Karol \.{Z}yczkowski and Wojeciech S{\l}omczy\'{n}ski.
\newblock The {M}onge distance between quantum states.
\newblock {\em J. Phys. A}, 31(45):9095--9104, 1998.
\newblock \href {https://doi.org/10.1088/0305-4470/31/45/009} {\path{doi:10.1088/0305-4470/31/45/009}}.

\end{thebibliography}

\end{document}